\documentclass[10pt,journal,compsoc]{IEEEtran}
\usepackage{tikz}
\ifCLASSOPTIONcompsoc
  \usepackage[nocompress]{cite}
\else
  \usepackage{cite}
\fi

\usepackage{amsmath}
\usepackage{color}
\usepackage{enumerate}
\usepackage{graphicx}
\usepackage{makecell}
\usepackage{textcomp}
\usepackage{enumitem}

\usepackage{algorithm}
\usepackage{algorithmicx}
\usepackage{algpseudocode}

\usepackage{amsthm,amssymb,amsfonts}
\usepackage[pdfencoding=auto]{hyperref}  

\newtheorem{theorem}{Theorem}[section]

\newtheorem{lemma}[theorem]{Lemma}

\newtheorem{definition}[theorem]{Definition}

\usepackage{booktabs} 
\usepackage{multirow}
\usepackage{tabularx}
\usepackage{setspace}
\definecolor{color_speedup}{RGB}{20, 20, 20}
\newcommand{\tspeedup}[1]{{\scriptsize\color{color_speedup}(#1\ensuremath{\times})}}

\newcommand{\etal}{\textit{et al.}}
\newcommand{\lmart}{\textsc{$\lambda$MART}}

\newcommand{\dbert}{\textsc{duoBERT}}

\newcommand{\A}{\mathcal{A}}
\newcommand{\D}{\mathcal{D}}
\newcommand{\Ci}{\mathcal{C}}
\renewcommand{\l}{\left}
\renewcommand{\r}{\right}

\begin{document}

\title{An Optimal Algorithm for Finding Champions\\in Tournament Graphs}

\author{
  Lorenzo Beretta, Franco Maria Nardini, Roberto Trani, and Rossano Venturini
  \IEEEcompsocitemizethanks{
    \IEEEcompsocthanksitem{ Lorenzo Beretta is with the Basic Algorithms Research Copenhagen (BARC), University of Copenhagen. E-mail:\href{beretta@di.ku.dk}{beretta@di.ku.dk}}
 
    \IEEEcompsocthanksitem{ Franco Maria Nardini and Roberto Trani are with the National Research Council of Italy. E-mail: \{\href{francomaria.nardini@isti.cnr.it}{francomaria.nardini}, \href{roberto.trani@isti.cnr.it}{roberto.trani}\}@isti.cnr.it}

    \IEEEcompsocthanksitem{ Rossano Venturini is with the Department of Computer Science, University of Pisa. E-mail:\href{rossano.venturini@unipi.it}{rossano.venturini@unipi.it}}
    }

  \thanks{
  This paper extends a previous contribution by Beretta \emph{et al.}~\cite{BerettaNTV19}.
  This work is supported by the ``Algorithms, Data Structures and Combinatorics for Machine Learning'' (MIUR-PRIN 2017) and  PNRR ECS00000017 Tuscany Health Ecosystem Spoke 6 ``Precision medicine \& personalized healthcare'', funded by the European Commission under the NextGeneration EU programme.
  }
}

\IEEEtitleabstractindextext{


\begin{abstract}
A tournament graph is a complete directed graph, which can be used to model a round-robin tournament between $n$ players.
In this paper, we address the problem of finding a champion of the tournament, also known as Copeland winner, which is a player that wins the highest number of matches.
In detail, we aim to investigate algorithms that find the champion by playing a low number of matches.
Solving this problem allows us to speed up several Information Retrieval and Recommender System applications, including question answering, conversational search, etc.
Indeed, these applications often search for the champion inducing a round-robin tournament among the players by employing a machine learning model to estimate who wins each pairwise comparison.
Our contribution, thus, allows finding the champion by performing a low number of model inferences.
We prove that any deterministic or randomized algorithm finding a champion with constant success probability requires $\Omega(\ell n)$ comparisons, where $\ell$ is the number of matches lost by the champion.
We then present an asymptotically-optimal deterministic algorithm matching this lower bound without knowing $\ell$, and we extend our analysis to three variants of the problem.
Lastly, we conduct a comprehensive experimental assessment of the proposed algorithms on a question answering task on public data.
Results show that our proposed algorithms speed up the retrieval of the champion up to $13\times$ with respect to the state-of-the-art algorithm that perform the full tournament.

\end{abstract}

\begin{IEEEkeywords}
  Tournament Graph, Round-Robin Tournament, Copeland Winner, Minimum Selection, Pairwise Ranking.
\end{IEEEkeywords}
}

\maketitle

\begin{tikzpicture}[remember picture,overlay]
	\node[anchor=south,yshift=10pt] at (current page.south)
	{\fbox{\parbox{\dimexpr\textwidth-\fboxsep-\fboxrule\relax}{
				\footnotesize{
					\copyright 2023 IEEE. Personal use of this material is permitted.
					Permission from IEEE must be obtained for all other uses, in any
					current or future media, including reprinting/republishing this
					material for advertising or promotional purposes, creating new
					collective works, for resale or redistribution to servers or lists, or
					reuse of any copyrighted component of this work in other works.
				}
	}}};
\end{tikzpicture}

\IEEEpeerreviewmaketitle


\section{Introduction}
\label{bf:sec:intro}
A \emph{tournament graph} is a complete directed graph $T = \left(V, E \right)$, where $V$ and $E$ are the sets of nodes and arcs, respectively~\cite{GS:reid2013tournaments}.
The tournament graph can be used to model a round-robin tournament between $n$ players, where each player plays a match with any other player.
The orientation of an arc tells the winner of the match, i.e., we have the arc $(u, v) \in E$ iff $u$ beats $v$ in their match. In the following, we call \emph{arc lookup} or \emph{arc unfold} the operation of looking at the direction of an arc between two nodes.

We address the problem of finding a champion of the tournament, also known as \textit{Copeland winner}~\cite{GS:copeland1951reasonable}, which is a vertex in $V$ with the maximum out-degree, i.e., a player that wins the highest number of matches.
Our goal is to find a champion by minimizing the number of arc lookups, i.e., the number of matches played.
Note that a tournament graph may have more than one champion.
In this case, we aim at finding any of them, even if all the proposed algorithms are able to find all of them without increasing the complexity.

If the tournament is transitive---whenever $u$ wins against $v$ and $v$ wins against $w$, then $u$ wins against $w$---we can trivially identify the unique tournament champion with $\Theta(n)$ arc lookups.
Indeed, the champion is the only vertex that wins all its matches and, thus, we can perform a \textit{knock-out tournament} where the loser of any match is immediately eliminated.
However, finding the champion of general tournament graphs requires $\Omega(n^2)$ arc lookups~\cite{DBLP:journals/corr/GutinMR18}, and thus, there is nothing better to do than to play all the matches.
This means that the structure of the underlying tournament graph heavily impacts the complexity of the problem.

In this article, we parametrize the problem with the number $\ell$ of matches lost by the champion and we investigate efficient algorithms that find the champion by performing a number of arc lookups proportional to $\ell$.
This parametrization is motivated by many applications in Information Retrieval and Recommender Systems that exploit pairwise machine learning (ML) models. These models compare a pair of candidate players at a time to estimate who wins the match. The final champion of the tournament is the player winning the highest number of pairwise comparisons of the all-vs-all tournament induced by the machine-learned model~\cite{GS:ai2019learning,DBLP:journals/corr/abs-1901-04085}.
The parametrization we introduce is motivated by the fact that, nowadays, it is possible to design accurate pairwise models that achieve a low error rate in the estimation of the matches played by the champion.
For this reason, we expect a low number of matches $\ell$ lost by the champion, hence a quasi-linear number of arc lookups is required by our algorithms to find it.
This compares with the quadratic number of lookups needed by the previously known algorithms~\cite{DBLP:journals/corr/GutinMR18}.
For this reason, this paper proposes efficient algorithms to find the tournament champion by performing the (asymptotically) minimum number of calls to the machine learning model, i.e., arc lookups, needed to solve this problem.
A more detailed description of the application scenarios is reported at the end of this section.

\subsubsection*{Our Contributions} The novel contributions of this article are the following:
\begin{itemize}[leftmargin=*]
	\item we introduce an asymptotically-optimal deterministic algorithm that finds the champion by employing $O(\ell n)$ vertex comparisons, where $\ell$ is the minimum number of matches lost by any player. Moreover, we prove that $\Omega(\ell n)$ comparisons are necessary, even for randomized algorithms, to obtain a correct answer with any constant probability. It is worth noticing that we match a randomized lower bound with a deterministic algorithm, showing that randomization does not give any advantage to this problem.
	
	\item We extend our result to three strictly-related problems. First, we show how to retrieve all top-$k$ players in time $O(\ell_k n)$, where $\ell_k$ is the number of matches lost by the $k$-th best player. Second, we consider a model of computation in which we are allowed to play a batch of $B$ matches in parallel, and we design an algorithm that achieves optimal speedup with respect to the sequential version and it finds the champion by performing $O(\frac{\ell n}{B} + \ell \log B)$ arc lookups. This is useful in practice because pairwise comparisons can be batched when the inference is done on novel computing platforms like, for example, GPUs. 
	Third, we generalize the tournament problem in a probabilistic framework, where each arc $(u, v) \in E$ is labeled with the likelihood that $u$ wins against $v$.
	These probabilities can be interpreted as the confidence of the machine learning model about the outcome of the comparison. In this setting, we define the champion as the player that minimizes the expected number of matches lost and we introduce an algorithm to find all champions in time $\Theta(\ell n)$, where $\ell$ is the expected number of matches lost by the champion.
		
	\item We provide a comprehensive experimental assessment of the proposed algorithms. We evaluate their performance in terms of running time and number of comparisons against a baseline that perform all the possible pairwise comparisons between players. We focus our attention on a Question Answering task that asks to find the most relevant textual answer to a given question provided by a user~\cite{DBLP:journals/corr/abs-1910-14424}. Results show that our proposed algorithms allow us to speed up the identification of the correct answer of up to $13\times$ with respect to methods that play the full tournament.
\end{itemize}

\subsubsection*{Application Scenarios}
Our investigation is motivated by many application scenarios involving the efficient selection of the most relevant result from a pool of candidates, also known as top-$1$ retrieval.
It is a crucial task in many Information Retrieval and Recommender System applications including Web ad-hoc search~\cite{baeza1999modern}, question answering~\cite{DBLP:journals/tois/HerlockerKTR04}, conversational search~\cite{DBLP:conf/cikm/Maarek18}, etc.
A recent example in this line is conversational assistants.
These devices, such as Siri, Google Assistant, and Alexa, are becoming very popular nowadays.
They work by exploiting a new way of interaction with the user, where the latter interacts by asking a question and the former provides her the answer with the highest relevance with respect to the question.
Conversational assistants introduce a \emph{paradigm shift} in information retrieval as they change the way users submit their information needs to the information retrieval system, i.e., using spoken words and not textual queries.
Moreover, since the new paradigm employs a conversation as a means of interaction, only one result is provided to the user as an answer to her question.
As a consequence, the precision in the identification of the only answer to return is now of paramount importance to build an effective conversational system.

State-of-the-art solutions for solving the top-$1$ retrieval task rely on machine learning techniques~\cite{DBLP:journals/ftir/Liu09}, to select the answer with the highest relevance.
The selection of the most relevant result can be addressed in two different ways:
\textit{i}) by exploiting machine-learned techniques such as \lmart~\cite{DBLP:journals/ir/WuBSG10}, which are based on univariate scoring functions that individually estimate one candidate result at a time, to select the candidate achieving the highest relevance score;
\textit{ii}) by employing pairwise Learning-to-Rank techniques such as \dbert{}~\cite{DBLP:journals/corr/abs-1901-04085}, which are based on bivariate scoring functions that estimate a pair of candidate results at a time, e.g., a binary judgment stating which of the results is more relevant, to select the candidate achieving the highest sum of pairwise scores of an all-vs-all tournament.
While the former approach exploits only the information of a single result at a time for computing the ranking score, the latter approach is potentially more powerful because it exploits the information of two candidates at a time for computing the outcome of the tournament.
However, the latter approach, although effective, is more expensive than the former one as it performs a quadratic number of comparisons to score all pairs of candidate results, thus making pairwise approaches unappealing in scenarios with tight time constraints. Here is where our research is beneficial as we define algorithmic approaches that allow reducing the number of comparisons performed by the pairwise model to select the most relevant results thus speeding up the whole selection process.

\smallskip
The rest of the article is structured as follows: Section~\ref{bf:sec:related} discusses the related work while Section~\ref{bf:sec:complexity} provides a detailed analysis of the problem complexity, and Section~\ref{bf:sec:algorithm} presents an efficient algorithm to solve it. Moreover, Section~\ref{bf:sec:generalizations} discusses three variants of the algorithm that solve three extensions of the original problem. Finally, Section~\ref{bf:sec:experiments} presents a comprehensive analysis of proposed algorithms in a information retrieval (ad-hoc search) scenario, and Section~\ref{bf:sec:summary} concludes the work.


\section{Related Work}
\label{bf:sec:related}
Tournament graphs are a well-known model that has been applied to several different areas such as sociology, psychology, statistics, and computer science.  Examples of applications are round-robin tournaments, paired-comparison experiments, majority voting, communication networks, etc.~\cite{DBLP:reference/choice/BrandtBH16,GS:reid2013tournaments,DBLP:journals/mss/Hudry09,GS:laslier1997tournament,GS:moon1968topics}.
In this area, we identify two different research lines. The first one aims at finding \emph{the tournament winner}, while the second one aims at \emph{ranking the list of candidates} using pairwise approaches.
Given a ranking of candidates, we can easily define the champion as the top-$1$ element of a the global ranking, therefore the two tasks are related with each other.
In this section, we describe the most important results concerning these two problems.

According to previous works~\cite{DBLP:reference/choice/BrandtBH16,GS:laslier1997tournament,GS:moon1968topics}, there is no unique definition of the notion of a tournament winner. Nevertheless, all of them agree on defining the winner whenever there
is a candidate, called \textit{Condorcet winner}, which beats all the others. Different definitions of winner require different complexities of the algorithms used to identify it. The easiest case to consider appears when $T$ is a \textit{transitive tournament graph}, i.e., a directed acyclic graph, since it is trivial to find the Condorcet winner in linear time by performing a knock-out tournament where the loser of any match is immediately eliminated. Instead, for a general tournament $T$, the complexity of finding a winner is much higher and strictly depends on the definition of winner.

A winner as defined by Banks~\cite{GS:banks1985sophisticated} is the Condorcet winner of a maximal transitive sub-tournament of $T$. As there may be several of these sub-tournaments, the \textit{Banks solution} is the set of all these winners. The problem of finding just one winner can be computed in $\Theta(n^2)$ arc lookups, while finding all of them is a $\mathcal{NP}$-hard problem~\cite{DBLP:journals/mss/Hudry09}.

Slater~\cite{GS:slater1961inconsistencies} defined the winner starting from a ranking of candidates. He defined a \textit{Slater solution} to be a total order $\prec$ on vertices that minimizes the number of mis-ordered pairs of vertices, where a pair $(u, v)$ is mis-ordered if $u$ beats $ v$ and $u \prec v$. The champion is then defined as the maximum element with respect to $\prec$. However, the computation of the Slater solution is $\mathcal{NP}$-hard as it reduces from the \textit{Feedback Arc Set Problem}~\cite{DBLP:journals/cpc/CharbitTY07}.

Ailon \emph{et al.}~\cite{DBLP:journals/ml/AilonM10,DBLP:conf/colt/AilonM08} provide a bound to the error achieved by the Quicksort algorithm when used to approximate a Slater solution. The error is defined as the number of misordered pairs of vertices. Ailon \emph{et al.} show that the expected error is at most two times the best possible error. It is apparent that the proposed algorithm requires $\Omega(n \log n)$ arc lookups with high probability. Even though the overall approximation is good, this algorithm fails in finding a champion $w$ every time one of the Quicksort pivots beats $w$, hence it is not suitable for our purposes.

The results by Shen \emph{et al.}~\cite{DBLP:journals/siamcomp/ShenSW03} and Ajtai \emph{et al.}~\cite{DBLP:journals/talg/AjtaiFHN16} provide a ranking based on the definition of \textit{king}.
The vertex $u$ is a king if for every vertex $v$ there is a directed path from $u$ to $v$ of length at most $2$ in $T$.
The ranking algorithm by Jian \emph{et al.}~\cite{DBLP:journals/siamcomp/ShenSW03} finds a sorted sequences of vertices $u_1, u_2, \ldots, u_n$ such that for every $i$ 1) $u_i$ beast $u_{i+1}$, and 2) $u_i$ is a king in the sub-tournament induced by the items $u_{i}, u_{i+1}, \ldots, u_n$.
The authors provide a $O(n^{3/2})$ deterministic algorithm to compute this sequence. On the flip side, a $\Omega(n^{4/3})$ deterministic lower bound for the retrieval of a single king holds. In addition, quicksort produces such a sequence in $O(n \log n)$ comparisons w.h.p. and quickselect retrieves a king in expected linear time. 
To date the deterministic complexity of finding a king in a tournament is still unknown, however attempts at understanding the problem proceed relentless~\cite{kings-newest}.
Unfortunately, the definition of king is weaker than the one of Copeland winner. Indeed, the latter implies the former~\cite{GS:reid2013tournaments}, and it is possible to construct tournaments in which every vertex is a king. Thus the definition of king does not help us in the identification of the best candidate.

A prolific research line studies the ranking problem under \emph{persistent comparison errors}
\cite{DBLP:conf/esa/Geissmann0LP19,DBLP:conf/isaac/Geissmann0LP17,DBLP:conf/soda/BravermanM08,DBLP:conf/esa/KleinPSW11}.
This task deals with queries affected by random noise in a scenario where comparison errors are persistent.
In this setting, we consider the set of vertices as equipped with a transitive order $\prec$, and every arc of the tournament as the result of a noisy comparison between two items.
The answer associated to the comparison $(u,v)$ is consistent with the transitive order $\prec$ with probability $p \approx 1$ and inconsistent with probability $1-p \approx 0$. All comparisons are independent.
By defining the dislocation of $u$ as the difference between its real rank and the rank assigned by an algorithm, Geissmann \emph{et al.}~\cite{DBLP:conf/isaac/Geissmann0LP17} proved that every algorithm produces a ranking with maximum dislocation $\Omega(\log n)$ and total dislocation $\Omega(n)$.
A recent work by Geissmann \emph{et al.}~\cite{DBLP:conf/esa/Geissmann0LP19} settles the problem, matching both lower bounds in $O(n \log n)$ time.
Unfortunately, this model does not produce a strong enough guarantee on the quality of the champion, that is only known to be within the top $O(\log n)$ candidates of the original ranking.

A line of work on \emph{non-persistent comparison errors} studies noisy comparisons under the assumption that every comparison can be queried more than once and the results are all independent. Recently, progress has been made on approximate selection~\cite{approximate-selection-newest}, and more notably on minimum-selection~\cite{minimum-selection-newest} that is exactly the problem we tackle in this paper, with a different model for noise. In fact, Leucci and Liu~\cite{minimum-selection-newest} just settled the complexity of minimum-selection in the non-persistent comparison error model.

There are several other notions of winner, and most of them can be computed in polynomial time. We refer to Hudry~\cite{DBLP:journals/mss/Hudry09} for a complete survey on this topic.
The definition used in this paper is the one given by Copeland~\cite{GS:copeland1951reasonable}, called \textit{Copeland solution}, where we rank vertices according to the number of matches they win, and a champion is the candidate winning the most matches.
As we already mentioned, the Copeland solution requires $\Omega(n^2)$ arc lookups and there is a trivial algorithm to match it~\cite{DBLP:journals/corr/GutinMR18}.
However, Geissmann \emph{et al.}~\cite{DBLP:conf/fct/GeissmannMW15} considered a model, similar to the aforementioned persistent comparison errors model, in which errors are no longer stochastic but their total number is bounded.
They fix an upper-bound $e$ to the total number of errors and they propose an algorithm to find the Copeland winner of the resulting tournament in $O(n \sqrt{e})$ comparisons and time.

\subsubsection*{Advancements over Previous Work}
In this article, we advance the state of the art by reporting improvements over the result by Geissmann \emph{et al.}~\cite{DBLP:conf/fct/GeissmannMW15}.
In particular, we propose an algorithm that finds the Copeland winner in $\Theta(\ell n)$ time and comparisons, where $\ell$ is the minimum number of matches lost by any player, hence $\ell \leq \sqrt{e}$ meaning that our algorithm is at least asymptotically as fast as Geissmann \emph{et al.}~\cite{DBLP:conf/fct/GeissmannMW15}. It is worth noting that in our use case $\ell$ is very small, and so this parameterization is particularly insightful.
Moreover, our novel algorithm presented in Section \ref{bf:sec:algorithm} is oblivious with respect to $\ell$, while the algorithm by Geissmann \emph{et al.}~\cite{DBLP:conf/fct/GeissmannMW15} assume to know $e$ in advance.
Finally, we provide a \emph{randomized} lower bound that matches the complexity of our \emph{deterministic} and simple algorithm (Section~\ref{bf:sec:randomized-lb}).
One last remarkable contribution is the extension of our algorithm to work when comparisons can be performed in batches and we achieve virtually no asymptotic overhead with respect to perfect parallelism (Section~\ref{bf:sec:batched}).

\section{Lower bounds}
\label{bf:sec:complexity}
In this section, we prove the lower bound of the Copeland winner problem.
An adversarial argument is used by Gutin \emph{et al.} \cite{DBLP:journals/corr/GutinMR18} to prove that finding a champion requires $\Omega\left(n^2\right)$ arc lookups.
Therefore, the trivial algorithm that finds a champion by performing all the possible matches
is optimal in general.
The problem is indeed much more interesting if we parameterize it with $\ell$, the number of matches lost by the champion.
Note that $\ell$ is unknown to the algorithm.
The goal of this section is to prove that $\Omega(\ell n)$ arc lookups are necessary to find a champion.
We first show that this bound applies to deterministic algorithms. Then we generalize it to the class of ``Monte Carlo'' randomized algorithms that are allowed to return an incorrect answer with a fixed positive probability.
The latter result clearly implies the former.
However, for pedagogical reasons we report them in increasing order of difficulty.

\subsection{Deterministic Lower Bound}
The following theorem shows that any deterministic algorithm employs $\Omega(\ell n)$ arc lookups to find a champion.

\begin{theorem}
\label{bf:alg:lower_bound}
Any deterministic algorithm that finds a champion in a tournament graph $T$ with $n$ vertices and with $\ell$ matches lost by the champion requires $\Omega(\ell n)$ arc lookups.
\end{theorem}

\begin{proof}
The lower bound is proved by using an adversarial argument. Assume that an algorithm claims that a vertex $u$, losing $\ell$ matches, is a champion by performing $\frac{1}{2}\ell (n-1)$ arc lookups.
There must exist a node $v$ such that the algorithm has performed less than $\ell$ lookups of arcs incident to $v$.
We thus can let the algorithm be incorrect by adversarially setting $v$ as the winner of those matches, so that $v$ wins more matches than $u$.
In other words any correct algorithm, claiming that a vertex $u$ is a champion with $\ell$ matches lost, must be able to certificate its answer by showing: 1) a list of $n - 1 - \ell$ matches won by $u$ and 2) a list of $\ell$ matches lost by any other vertex $v$.
\end{proof}

\subsection{Randomized Lower Bound}
\label{bf:sec:randomized-lb}
We just proved that no deterministic algorithm can perform $o(\ell n)$ arc lookups and output a correctness certificate.
Now we extend such a non-existence result to any randomized algorithm, which is allowed to be wrong with a fixed probability.
This section is devoted to prove the following theorem stating that it does not exist a Monte Carlo
algorithm that finds the Copeland winner with $o(\ell n)$ arc lookups.

\begin{theorem}
\label{monte_carlo_thm}
Given a tournament $T$ with $n$ vertices and with $\ell$ matches lost by the champion, it does not exist a randomized algorithm that performs $o(\ell n)$ arc lookups and outputs the Copeland winner of $T$ with fixed positive probability.
\end{theorem}

To prove the theorem above, we need to define the auxiliary problem below and operate a reduction.
\begin{definition}[Anomalous Row Problem]
\label{anomalous_row_problem}
Given a matrix $M \in \mathbb{F}_2^{k \times m}$ such that every row but one presents $k + 1$ zeroes and the remaining one presents $k$ zeroes, find the $k$-zeroes row.
\end{definition}

We will see that the anomalous row problem is not harder than the problem of finding the Copeland winner: technically we will show a reduction between these two problems. Moreover, proving a randomized lower bound for the anomalous row problem turns out to be easier.

The next lemma bounds from below the number of $M$'s entries that must be probed in order to solve the anomalous row problem. This bound is strictly related to Theorem~\ref{monte_carlo_thm}, as we will see shortly.
\begin{lemma}\label{anomalous_thm}
It does not exist a randomized algorithm that solves the anomalous row problem (Definition~\ref{anomalous_row_problem}) by probing $o(km)$ cells of the input matrix $M$ and returns the correct answer with fixed positive probability.
\end{lemma}

To ease the discussion, we defer the proof of Lemma~\ref{anomalous_thm} to the end of this section. First, we show that if there exists an algorithm violating Theorem~\ref{monte_carlo_thm} then we can design an algorithm that violates Lemma~\ref{anomalous_thm}.
Thus, proving Lemma~\ref{anomalous_thm} is sufficient to prove Theorem~\ref{monte_carlo_thm}.

Given an instance of the anomalous row problem, $M \in \mathbb{F}_2^{k \times m}$, we can assume that $k$ and $m$ are odd and $m > 3k$.
Indeed, if this is not the case, it is sufficient to add a dummy row containing $k + 1$ zeroes and several dummy columns containing only ones.
It is apparent that this modification preserves both the $k$-zeroes row and the asymptotic complexity.
Then, we construct a tournament having $n = k + m$ players and adjacency matrix

\begin{equation*}
  A =
  \left[
\begin{array}{cc}
  B & M \\
  \widetilde{M}^T & C
\end{array}
\right] \in \mathbb{F}_2^{n \times n}
\end{equation*}
where $B \in \mathbb{F}_2^{k \times k}$ and
$C \in \mathbb{F}_2^{m \times m}$ are the adjacency matrices of regular tournaments\footnote{A $(2j + 1)$-vertices tournament is said to be \textit{regular} if every vertex has out-degree $j$.} and
$\widetilde{M}$ is the \emph{complementary matrix} of $M$, meaning that $\widetilde{M}_{i, j} = 1 - M_{i, j}$.

We can easily prove that the champion is among the first $k$ players and loses exactly $\ell = (3k - 1) / 2$ matches.
In fact, due to regularity, every row in $B$ contains exactly $(k - 1) / 2$ zeroes and $M$ satisfies the hypotheses of Definition~\ref{anomalous_row_problem}.
Thus, every player among the first $k$ ones loses either $\ell$ or $\ell + 1$ matches.
On the other hand, any player in the last $m$ rows, loses at least $(m - 1) / 2$ matches, and
$m > 3k$ guarantees that $(m - 1) / 2 > \ell$.
Therefore, if we find a champion of the constructed tournament then we automatically solve the
anomalous row problem.

We are now left to prove Lemma~\ref{anomalous_thm}.
First we enunciate a game-theoretic lemma by Yao \cite{DBLP:conf/focs/Yao77} declined within the terms of our problem.

\begin{lemma}[Yao's Lemma]
Let $\A$ be the family of \textit{deterministic} algorithms that output a, possibly wrong, solution to the anomalous row problem and probe $o(km)$ cells.
Consider $\A$ equipped with a probability distribution.
Then consider the function $\Ci(A, x)$ that returns $1$ if the algorithm $A$ is correct on input $x$ and $0$ otherwise. Finally, consider a probability distribution $\D$ over $\mathbb{F}_2^{k \times m}$. We have

\begin{equation*}
	\min_{x \in \D}\mathbb{E}_{\A}\l[\Ci(A, x)\r] \leq
	\mathbb{E}_{\A \otimes \D}\l[\Ci(A, x)\r] \leq
	\max_{A\in \A} \mathbb{E}_{\D}\l[\Ci(A, x)\r].
\end{equation*}
\end{lemma}

We know that a Monte Carlo algorithm that proves $o(kn)$ cells can be represented as a probability distribution $\A$, in fact it just tosses some coins at run-time and it decides which algorithm to branch into.
Therefore $\min_{x \in \D}\mathbb{E}_{\A}\l[\Ci(A, x)\r]$ is the probability of the Monte Carlo algorithm defined by $\A$ of being right in the worst case, and $\max_{A\in \A} \mathbb{E}_{\D}\l[\Ci(A, x)\r]$ is the average case of the best deterministic algorithm against a random input with distribution $\D$. Finally, we prove Lemma~\ref{anomalous_thm}.

\begin{proof}[Proof of Lemma~\ref{anomalous_thm}]
It is sufficient to show an input distribution $\D$ such that any deterministic algorithm with running time $o(km)$ succeeds with arbitrarily small probability, for $k, m \rightarrow \infty$.
We choose the permutation $\phi $ of $\{1 \dots k\}$ and $k$ permutations $\sigma_1 \dots \sigma_k$ of $\{1 \dots m\}$ uniformly at random.
Consider the random matrix $X \in \mathbb{F}_2^{k \times m}$ such that
\begin{equation*}
	X\l[i, j\r] = M\l[\phi(i), \sigma_i(j)\r]
\end{equation*}
where $M$ is a deterministic input matrix as in Definition~\ref{anomalous_row_problem}.
Let $\D$ be the distribution of $X$, and $A \in \A$ be the algorithm such that $\mathbb{E}_{\D}\l[\Ci(A, x)\r]$ is maximum.
It is sufficient to show that $\mathbb{E}_{\D}\l[\Ci(A, x)\r] \rightarrow 0$ to prove that no Monte Carlo algorithm can perform less than $\Omega\l(km\r)$ cell probes.
Consider the maximum number $P$ of cells probed by $A$ and define $$\Gamma_{k, m} = \min\l(\sqrt{\frac{km}{P}}, \,k\r).$$
We now color $\Gamma_{k, m}$ cells in the input matrix.
We first color a $1$-valued cell in the $k$-zeroes row, then we choose $\Gamma_{k, m} - 1$ rows containing $k + 1$ zeroes and color a $0$-valued cell drawn from each of those.
To this end, we assume to perform such coloring before randomizing the input.
We want to estimate the probability that the algorithm probes any colorful cell.
Define the event $E_i$ ``the algorithm picks a colorfull cell during the $i$-th probe''.
The probability of $E_i$ is $\Gamma_{k, m} / km$ since the chosen
cell's row contains a colorful cell with probability $\Gamma_{k, m} /
k$ and, given that
, the probability of picking the colorful cell is $1 / m$.

Therefore,

\begin{equation*}
  \mathbb{P}_{\D}\l(\bigcup_{i = 1}^P E_i\r)  \leq
  \sum_{i = 0}^P \mathbb{P}_{\D}\l(E_i\r)  \leq
  \frac{P\Gamma_{k, m}}{km}  \leq
  \sqrt{\frac{P}{km}}.
\end{equation*}
Finally, we notice that, in case none of the colorful cells is probed, the algorithm ``sees'' a perfectly symmetric distribution over the $\Gamma_{k, n}$ rows containing a colorful cell.
Therefore, the best it can do is to produce a random output, which is right with probability $1 / \Gamma_{k, n}$, at most. To conclude, consider $E = \bigcup_{i = 1}^{P} E_i$, we have

\begin{equation*}
  \begin{split}
    \mathbb{E}_{\D}\l[\Ci(A, x)\r] & \leq
    \mathbb{P}_{\D}\l(E\r) + \mathbb{P}_{\D}\l(\Ci(A, x) = 1 \,|\, E^c \r) \\
    & \leq \frac{P\Gamma_{k, n}}{kn} + \frac{1}{\Gamma_{k, n}}
    \rightarrow 0
  \end{split}
\end{equation*}
where the last limit holds for $k$ and $n$ that goes to infinity simultaneously.
\end{proof}


\section{Optimal deterministic algorithm}
\label{bf:sec:algorithm}

In this section, we present a simple, deterministic, and asymptotically optimal algorithm that finds every champion in $\Theta(\ell n)$ arc lookups and time.
We first introduce the algorithm.
Then, we prove its correctness and we bound the number of arc it lookups.
Finally, we discuss some implementation details to show that the number of operations performed by the algorithm is $\Theta(\ell n)$ and the space required is linear.

\subsection{Algorithm Description}
We detail our algorithm in Algorithm~\ref{bf:alg:pseudocode}.
The number $\ell$ of matches lost by the champion is unknown to the algorithm.
Thus, it performs an exponential search to find the suitable value of $\alpha$ such that $\alpha/2 \leq \ell < \alpha$ (line~\ref{bf:alg:pseudocode_outer}) so to solve the problem by assuming that the champion loses less than $\alpha$ matches.

\begin{algorithm}[htb]
  \caption{}\label{bf:alg:pseudocode}
  \begin{algorithmic}[1]
    \Procedure{FindChampion}{$T = \left(V, E\right)$}
      \For {($ \alpha = 1 $; true; $\alpha = 2 \alpha$)} \label{bf:alg:pseudocode_outer}
        \State {$ A = V $}
        \State {$ S = \left\{(u,u) \mid u \in V\right\} $}
        \State {$ \forall u \in V \;\;lost[u] = 0 $} 
        \While {$ |A| > 2 \alpha$} \label{bf:alg:pseudocode_inner}
          \State {choose a pair of vertices $u, v$ in $A^2\setminus S$} \label{bf:alg:pseudocode_pairselection}
          \State {$ S = S \cup \left\{(u, v), (v, u)\right\}$} 
          \State {$loser =$} \algorithmicif\ {$(u,v) \in E$} \algorithmicthen\ {$v$} \algorithmicelse\ $u$
          \State {$\mathrel{++}\!lost[loser]$} \label{bf:increment}
          \If {$ lost[loser] \geq \alpha $}
            \State {$ A = A \setminus \left\{loser\right\} $} \label{bf:alg:pseudocode_kill}
          \EndIf
        \EndWhile
        \State {$ c, lost_c =$ \Call{FindChampionBruteForce}{$A$, $E$}}
        \If {$lost_c < \alpha $}
          \label{bf:alg:pseudocode_exitcondition}
          \Return {$c$}
        \EndIf
      \EndFor
    \EndProcedure
  \end{algorithmic}
\end{algorithm}

At each iteration, the algorithm maintains a set $A$ of ``alive'' vertices that is initially equal to $V$.
Then, it performs an elimination tournament among the vertices in $A$ by eliminating a player each time it loses $\alpha$ matches (line~\ref{bf:alg:pseudocode_kill}) until only $2\alpha$ vertices remain alive (line~\ref{bf:alg:pseudocode_inner}).
This stop condition guarantees the convergence of the algorithm.
The matches are selected arbitrarily to avoid to play the same match multiple times (line~\ref{bf:alg:pseudocode_pairselection}).
When the elimination tournament ends, a candidate champion is found via the \textsc{FindChampionBruteForce} procedure, which exhaustively finds
the vertex $c$ of $A$ with the maximum out-degree in $T$.
Whenever the candidate $c$ loses at least $\alpha$ matches (line~\ref{bf:alg:pseudocode_exitcondition}), the value of $\alpha$ is not the correct one and the champion may have been erroneously eliminated before.
Thus, $c$ could not be a champion and the algorithm continues with the next value of $\alpha$ (line~\ref{bf:alg:pseudocode_outer}).


In the reminder of this section, we prove the following theorem stating that Algorithm~\ref{bf:alg:pseudocode} matches the number of arc lookups indicated by the lower bound (Theorem~\ref{bf:alg:lower_bound}) and requires linear space.

\begin{theorem}
\label{theorem:algorithm_guarantees}
Given a tournament graph $T$ with $n$ vertices and with $\ell$ matches lost by the champion, Algorithm~\ref{bf:alg:pseudocode} finds every champion with $\Theta(\ell n)$ arc lookups and time. It also requires linear space.
\end{theorem}

\subsection{Correctness}
Let us first assume that the value of $\alpha$ is such that $\alpha/2 \leq \ell < \alpha$.
We now prove that, under this assumption, the algorithm correctly identifies a champion.
First, we observe that the algorithm cannot eliminate the champions as each of them loses less than $\alpha$ matches.
Thus, if we prove that the algorithm terminates, the set $A$ contains all the champions and the \textsc{FindChampionBruteForce} procedure will identify any
(potentially all) of them.
Note that a champion of $T$ may not be a champion of the sub-tournament restricted to only the vertices in $A$.
This is why \textsc{FindChampionBruteForce} procedure computes the out-degrees of all vertices in $A$ by looking at the edges of the original tournament $T$.
We use the following lemma to prove that, eventually, the condition $|A| = 2 \alpha$ is met and the algorithm terminates.



\begin{lemma}\label{lemma:vertex_exists}
In any tournament $T$ of $n$ vertices there is at least one vertex having in-degree $(n-1)/2$.
\end{lemma}

\begin{proof}
The sum of the in-degrees of all vertices of $T$ is exactly $\binom{n}{2} = \frac{n(n-1)}{2}$.
Since there are $n$ vertices, there must be at least one vertex with in-degree $\frac{n-1}{2}$.
\end{proof}

Thus, each tournament of $2\alpha + 1$ vertices, or more, has at least one vertex losing at least $\alpha$ matches.
This means that the algorithm has always the opportunity to eliminate a vertex from $A$ until there are $2 \alpha$ vertices left.
Notice that the above discussion is valid for any value of $\alpha$ smaller than the target one.
Thus, any iterations of the exponential search will terminate and it eventually finds a suitable value of $\alpha$, i.e., $\alpha/2 \leq \ell < \alpha$, where a champion will be identified.

\subsection{Complexity}
We now present an analysis of the complexity of the algorithm.
Let us first consider the cost of an iteration of the exponential search.
We observe that each arc lookup increases one entry of $lost$ by one and that none of these entries is ever greater than $\alpha$.
Thus, the elimination tournament takes no more than $n \alpha$ arc lookups.
Moreover, the \textsc{FindChampionBruteForce} procedure takes less than $2 n \alpha$ arc lookups since it just considers every arc of the remaining $2 \alpha$ alive nodes.
Thus, an iteration of the exponential search takes less than $3 n \alpha$ arc lookups.

We get the complexity of the overall algorithm by summing up over all the possible values of $\alpha$, which are all the powers of $2$ from $1$ up to $2\ell$.
Thus, we have at most $3n \sum_{i=0}^{\lceil \log_2 (2\ell) \rceil} 2^i=O(\ell n)$ arc lookups.

\subsection{Implementation Details}
\label{sec:implementation_details}
We now prove that Algorithm~\ref{bf:alg:pseudocode} can be implemented in $\Theta(\ell n)$ time and linear space.
We do this by exploiting the fact that Algorithm~\ref{bf:alg:pseudocode} allows us to choose any arc as soon as its vertices are alive and it has never looked up before.
An efficient implementation is achieved by maintaining two arrays of $n$ elements each: an array $A$ storing the alive vertices and an array $lost$ storing the number of matches lost by each vertex.
A counter $numAlive$ stores the number of alive vertices.
Our implementation maintains the invariant that the prefix $A[1,numAlive]$ contains only alive vertices.
We use two cursors $p_1$ and $p_2$ to iterate over the elements in $A$.
At the beginning $p_1 = 1$, $p_2 = 2$ and $numAlive = n$.
Our implementation performs a series of matches involving vertex $A[p_1]$ and all other vertices in $A[p_1+1, numAlive]$, thus, advancing the cursor $p_2$.
Then, it moves $p_1$ to the next position.
After every match between $A[p_1]$ and $A[p_2]$, we increment $lost$ of the loser, say vertex $v$.
Whenever $lost[v]$ equals $\alpha$, we eliminate $v$ according to the following two cases, then we decrease $numAlive$ by one.
The first case occurs when $v$ is $A[p_1]$.
We swap $A[p_1]$ and $A[numAlive]$, we end the current series of matches, and we start a new one.
The second case occurs when $v$ is $A[p_2]$. Here, we swap $A[p_2]$ and $A[numAlive]$, and we continue the current series of matches.
In both cases, we decrease $numAlive$ by $1$ so that we preserve the invariant.

A similar, slightly less efficient, implementation employs a linked list $A$ to store
the alive vertices.
In this implementation, the removal of an element from the list is trivial,
and $p_1$ and $p_2$ are pointers that always advance in the list.
When $p_2$ reaches the end of the list, we advance $p_1$ by one position in the
list and we set $p_2$ to point to the element just after $p_1$.
This implementation allows us to process the vertices according to the input order
(as we never swap elements), which may be desirable in practice if we can somehow
predict the strongest of the vertices and sort them according.

As each step of the exponential search ignores the arc lookups of the previous steps, i.e., certain arcs may be considered more than once.
Therefore, to reduce the number of arc lookups preserving the time complexity at the cost of using $\Theta(\ell n)$ space instead of $O(n)$, an hash table can be employed to store the results of all arc lookups across the exponential search steps so to avoid unnecessary repeated computations.
In detail, each time Algorithm~\ref{bf:alg:pseudocode} wants the result of a match, it checks the hash table first and, only if this is a new arc lookup, the algorithm compute the result of the match and stores the result in the hash table for the next exponential search steps.


\section{Generalizations of the Problem}
\label{bf:sec:generalizations}
We now discuss some generalizations of the Copeland winner problem and we modify Algorithm~\ref{bf:alg:pseudocode} to solve these problems efficiently.
First, we show how to retrieve the top $k$ items, i.e., not only the top-$1$, by maintaining the complexity proportional to the number of matches lost by the $k$-th player. Then, we consider the case of a binary machine learned classifier returning a pair of probabilities instead of a binary outcome and redefine the problem in a probabilistic fashion. Finally, we consider the case in which we are able to process batches of arc lookups in parallel, so to exploit parallel processing units, e.g., GPUs.

\subsection{Top-\ensuremath{k} retrieval Version}
A simple and useful generalization of the Copeland winner problem is to find the top-$k$ results, i.e., the $k$ vertices with the highest out-degrees.
In this setting, the exponential search of Algorithm \ref{bf:alg:pseudocode} can be modified to find the minimum value of $\alpha$ such that the number $\ell_k$ of matches lost by the $k$-th result is between $\alpha / 2$ and $\alpha$.
To this end, the exponential search must end whenever it finds $k$ vertices with less than $\alpha$ comparisons lost.
To accomplish this task, the \textsc{FindChampionBruteForce}$(A, E)$ procedure must be modified to return the indices of the top-$k$ results of $A$ along with number of matches lost by them.
Since $\ell_1 \leq \ell_2 \leq \ldots \leq \ell_n$, the higher the value of $k$, the higher the time complexity $O(n \ell_k)$ of the algorithm.

\subsection{Probabilistic Version}
Typically, the outcome of a pairwise classifiers is not a binary response, instead it is a pair of complementary probabilities that can be interpreted as the algorithm's confidence about the comparison's outcome.
Thus, a natural generalization of the Copeland winner problem emerges if we associate to each arc $(u, v)$ the probability $p_{u, v}$ of $u$ beating $v$.
Since the probabilities are complementary, we also know that $p_{v, u} = 1 - p_{u, v}$.
We refer to this graph as \emph{probabilistic tournament graph}. In this setting, the arcs are Bernoulli random variables, and we define the champion as the player $u$ minimizing the expected number of matches lost, i.e., $\sum_{v \in V} p_{v, u}$ by linearity of the expectation.
Since we want our complexity to be parameterized with the expected number of matches lost by the champion, we coherently call this quantity $\ell$.
In this section, we show that Algorithm~\ref{bf:alg:pseudocode} needs only little adaptation to work in this setting.

Consider the pseudocode of Algorithm~\ref{bf:alg:pseudocode}, we treat $lost$ counters as real-valued and substitute line~\ref{bf:increment} with two commands incrementing $lost[u]$ by $p_{v, u}$ and $lost[v]$ by $p_{u, v}$. Once operated these slight modifications we are ready to prove the following theorem (analogous of Theorem~\ref{theorem:algorithm_guarantees}).

\begin{theorem}
Let $T$ be a probabilistic tournament graph with $n$ vertices and with $\ell$ the expected number of matches lost by the champion. The modified version of Algorithm~\ref{bf:alg:pseudocode} described above finds every champion by requiring $\Theta(\ell n)$ arc lookups and time. The algorithm requires linear space.
\end{theorem}

\subsubsection*{Correctness}
The correctness proof is almost identical to the one we have detailed in Section~\ref{bf:sec:algorithm}. We are not repeating the whole proof, in fact it is sufficient to substitute occurrences of ``losses'' with ``expected losses'' and reformulate the Lemma~\ref{lemma:vertex_exists} as follows to obtain the desired proof.

\begin{lemma}
In any probabilistic tournament $T$ of $n$ vertices there is at least one vertex $u$ such that $\sum_{v \in V} p_{v, u} \geq (n-1)/2$. In other words, there exists a player whose expected number of matches lost is at least $(n-1)/2$.
\end{lemma}

\begin{proof}
The sum of the ``expected losses'' of all vertices of $T$ is exactly $\binom{n}{2} = \frac{n(n-1)}{2}$. Since there are $n$ vertices, there must be at least one vertex losing $\frac{n-1}{2}$ matches, on average.
\end{proof}

\subsubsection*{Complexity}
The complexity analysis is again akin to the one of Section~\ref{bf:sec:algorithm}, but we dig in a deeper details here. Each unfolded arc increases $\sum_{u \in V} lost[u]$ by one; since $lost[u]$ of any $u \in V$ is incremented until it surpasses $\alpha$ of at most one unit at a time, then $lost[u]$ cannot be greater than $\alpha + 1$ and $\sum_{u \in V} lost[u] < (\alpha + 1)n$.
Therefore no more than $(\alpha + 1 ) n$ arcs are unfolded during the elimination step of a single iteration of the exponential search.
Moreover, as in Section~\ref{bf:sec:algorithm}, \textsc{FindChampionBruteForce} procedure takes less than $2 n \alpha$ arc lookups. Thus an iteration of the exponential search takes less than $3 n (\alpha + 1)$ arc lookups, and summing up all these arc lookups we get the desired $O(\ell n)$ complexity.

\subsection{Parallel (Batched) Version}
\label{bf:sec:batched}
In modern architectures, e.g., GPUs, it is possible to perform multiple arc lookup operations in parallel.
A natural question is whether we are able to take full advantage of this parallelism to cut down the complexity of Algorithm~\ref{bf:alg:pseudocode}.
In this subsection, we propose Algorithm~\ref{batched:pseudocode} under the assumption to be able to unfold a batch of $B$ arcs in parallel.

In particular, Algorithm~\ref{batched:pseudocode} processes $O\l(\frac{\ell n}{B} + \ell \log B\r)$ batches, so the overhead is asymptotically negligible if $B = O\l(n / \log n\r)$, which is a condition that often holds in practice.

Algorithm~\ref{batched:pseudocode} is a slight modification of Algorithm~\ref{bf:alg:pseudocode}.
As the previous algorithm, it performs an exponential search of $\ell$ repeatedly doubling the parameter $\alpha$.
For each $\alpha$ it assumes that the champion belongs to the set of \textit{alive vertices} $A$ and performs an elimination tournament among the vertices of $A$ eliminating any player that loses $\alpha$ matches.
The elimination step is now performed in batches (line~\ref{batched:unfold}) and terminates when the alive players are few enough (line~\ref{batched:few_enough}).
The method \textsc{FindChampionBruteForce\textsubscript{Par}} (line~\ref{batched:brute_force}) can
be parallelized with no efforts by unfolding all $O\l(6 \alpha n\r)$ arcs in batches of $B$ arcs at a time, hence we focus on the elimination step.
The main difference with respect to Algorithm~\ref{bf:alg:pseudocode} resides in the procedure \textsc{BuildBatch}, which decides what are the $B$ arcs to lookup in parallel.
It creates local copies $A_{loc}$ and $lost_{loc}$ of the set $A$ and the vector $lost$, then the procedure selects matches in $A_{loc} \times A_{loc}$ that have not been played yet and, for each of them, assigns a loss to both opponents.
Now suppose that the batched games were played sequentially (namely, played at line~\ref{batched:push_back}) and $lost$ and $A$ were updated accordingly: we would have that $lost_{loc}$ provides an upper estimate of $lost$ and $A_{loc} \subseteq A$.
Therefore, it is guaranteed that every insertion in a batch will produce a match loss for a player that would be still alive in case we unfolded the batch sequentially.
This is a point worth stressing since it guarantees that $lost[u] \leq \alpha$ for each $u \in V$\footnote{We will employ this property in the complexity analysis.}.
Finally, even though it is not guaranteed that \textsc{BuildBatch} produces a $B$-sized batch, for that to happen it is sufficient that $A$ has at least $2B + 2\alpha$ elements.
This can be enforced halving the batch size every time this condition does not hold (line~\ref{batched:halving}) and we will see that this will not spoil the complexity of Algorithm
\ref{batched:pseudocode}.
Intuitively, the elimination step consists of two different epochs: the first one unfolds arcs in $B$-sized batches (where $B$ is the original batch size) until $|A| \geq 2B + 2\alpha$; the second one processes smaller and smaller batches until $|A|$ is small enough (line~\ref{batched:few_enough}).

\begin{algorithm}
	\caption{}\label{batched:pseudocode}
	\begin{algorithmic}[1]
\Procedure{FindChampion\textsubscript{Parallel}}{$T = \left(V, E\right)$, $B$}
	\For {($ \alpha = 1 $; true; $\alpha = 2 \alpha$)} \label{pseudocode_outer}
		\State {$ A = V $}
		\State {$ S = \left\{(u,u) \mid u \in V\right\} $}
		\State {$ \forall u \in V \;\; lost[u] = 0$} 
		\State {$ B' = B $}
		\While {$ |A| > 6 \alpha$} \label{batched:few_enough}
			\While {$|A| < 2B' + 2\alpha$} \label{batched:halving}
				\State {$B' = B' / 2$}
			\EndWhile
			\State {$batch =$ \Call{BuildBatch}{$A$, $S$, $B'$, $lost$, $\alpha$}} \label{buildBatch}
			\State {\Call{UnfoldInParallel}{$batch$}} \label{batched:unfold} 
			\For {$(u, v)$ in $batch$}
				\State {$loser =$} \algorithmicif\ {$(u,v) \in E$} \algorithmicthen\ {$v$} \algorithmicelse\ $u$
				\State {\Call{IncreaseLoss}{$A$, $lost$, $\alpha$, $loser$}}
			\EndFor \label{batched:endfor}
		\EndWhile
		\State {$c, lost_c =$\Call{FindChampionBruteForce\textsubscript{Par}}{$A$,$E$,$B$}}\label{batched:brute_force}
		\If {$lost_c < \alpha $} \label{pseudocode_exitcondition}
			\Return {$c$}
		\EndIf
	\EndFor
\EndProcedure
\State
\Procedure{BuildBatch}{$A$, $S$, $B^\prime$, $lost$, $\alpha$}
	\State {$batch = \varnothing$} \State {$A_{loc} = A$} \label{local_copy1} 
	\State {$lost_{loc} =lost$} \label{local_copy2} 
	\While {$|batch| < B^\prime$}
		\State {Choose $(u, v) \in A_{loc}^2 \setminus S$}
		\State {$S = S \cup \l\{(u, v), (v, u)\r\}$} 
		\State {$batch = batch \cup \l\{(u, v)\r\}$} \label{batched:push_back}
		\State {\Call{IncreaseLoss}{$A_{loc}$, $lost_{loc}$, $\alpha$, $u$}}
		\State {\Call{IncreaseLoss}{$A_{loc}$, $lost_{loc}$, $\alpha$, $v$}}
	\EndWhile
	\State {\textbf{return} $batch$}
\EndProcedure
\State{}
\Procedure {IncreaseLoss}{$A$, $lost$, $\alpha$, $v$}
	\State {$\mathrel{++}\!lost[v]$} 
	\If {$ lost[v] \geq \alpha $}
		\State {$ A = A \setminus \left\{v\right\}$} 
	\EndIf
\EndProcedure
	\end{algorithmic}
\end{algorithm}

\subsubsection*{Correctness}
The correctness can be proven exactly in the same way as the sequential case, the only detail to take care about is that the function \textsc{BuildBatch} terminates by producing a $B$-sized batch.
It is sufficient to notice that as long as $|A_{loc}| > 2\alpha$ there is an arc to unfold in $A_{loc}^2 \setminus S$ (using the same argument of the sequential case), and that since we call \textsc{IncreaseLoss} at most $2B'$ times at each iteration, then $|A| \geq 2B' + 2\alpha$ (line~\ref{batched:halving}) is sufficient to ensure the termination.

\subsubsection*{Complexity}
\begin{theorem}
Given a tournament graph $T$ with $n$ vertices and with $\ell$ 	matches lost by the champion, Algorithm~\ref{batched:pseudocode} finds every champion by requiring $O\l(\frac{\ell n}{B} + \ell \log B\r)$ calls to \textsc{UnfoldInParallel} and $O(\ell n)$ time and space.
\end{theorem}

\begin{proof}
Consider the $i$-th iteration of the cycle at line~\ref{batched:few_enough} and denote with $A_i$ the number of alive vertices $|A|$ and with $B_i$ the value of $B'$, evaluated immediately before calling the \textsc{BuildBatch} function at line~\ref{buildBatch}.
In particular we have $A_1 = |V|$ and $B_1$ = $B$. First, we prove the following lemmas.

\begin{lemma}
For each $i \geq 0$, $A_{i} - B_{i} \leq A_{i + 1} \leq A_i$ holds.
\end{lemma}

\begin{proof}
The first inequality holds since no more than $B_i$ games are played during the $i$-th iteration and thus no more than $B_i$ players are removed from the alive set. The second inequality holds since the set $A$ of alive vertices decreases over time.
\end{proof}

\begin{lemma}
\label{lemma_inequality}
Let $j$ be the first iteration in which the conditional statement at line~\ref{batched:halving} is true, that is $j = \min \l\{i \,|\, A_i < 2B + 2\alpha\r\}$. For each $i \geq j$, $2B_i + 2\alpha \leq A_i \leq 4 B_i + 2\alpha$ holds.
\end{lemma}

\begin{proof}
We prove it by induction.

\emph{Base Case, $i=1$}: it is sufficient to notice that $B_{j - 1} = 2 B_j$ and $A_j < 2B_{j - 1} + 2\alpha \leq A_{j - 1}$ hold thanks to the definition of $j$, and combine those equations with Lemma~\ref{lemma_inequality}.

\emph{Inductive Case, $i>1$}: during the $i$-th iteration we have two cases depending on whether we update the value of $B'$ or not.
If we do not update $B'$, that is $A_i \geq 2B_{i -1} + 2\alpha$, then we have $B_i = B_{i -1}$ and $A_i \leq A_{i -1} \leq 4B_{i -1} + 2\alpha$ by inductive hypotheses.
Otherwise, if we update $B'$, we have $A_i < 2B_{i - 1} + 2\alpha$ and $B_{i - 1} = 2 B_i$, then $A_i < 4B_i + 2\alpha$ and $A_i \geq A_{i - 1} - B_{i - 1} \geq B_{i - 1} + 2\alpha = 2B_i + 2\alpha$.
\end{proof}

We now fix $\alpha$ and upper-bound the number of arcs unfolded for each batch size $B_i$. First we deal with the case $B_i = B$ in which we employ the original batch size; in that case, we can safely upper-bound the number of arcs with $\alpha n$ since every $lost$ counter is never greater than $\alpha$ and every arc unfolded  increases a $lost$ counter by one.
Then, consider the case $B_i = B / 2^k$, for a specific value of $k$; we have that $|A_i| \leq 4 B / 2^k + 2\alpha$ and, by applying the same argument as above on the elements of $A_i$, that at most $\alpha \l(4 B / 2^k + 2\alpha\r)$ arcs are unfolded using a batch of size $B / 2^k$.
Thanks to the clauses at lines~\ref{batched:few_enough} and~\ref{batched:halving}, we have $6 \alpha \leq A_i \leq 4B_i + 2\alpha$, which implies $B_i \geq \alpha$.
To compute the total number of calls to \textsc{UnfoldInParallel}, it is sufficient to divide the maximum number of arc lookups ($\alpha |A_i|$) by the appropriate batch size ($B_i$) and sum them up

\begin{equation*}
\begin{split}
	\sum_i \frac{\alpha |A_i|}{B_i}
	& \leq \frac{\alpha n}{B} + \sum_{i=1}^{\lceil \log_2\l(B / \alpha\r) \rceil} \frac{\alpha |A_i|}{B_i}
	\\& \leq \frac{\alpha n}{B} + \sum_{i=1}^{\lceil \log_2\l(B / \alpha\r) \rceil} \frac{\alpha \l(4 (B / 2^i) + 2\alpha\r)}{B / 2^i}
	\\& < \frac{\alpha n}{B} + 4 \alpha \log_2 B + \frac{2 \alpha^2}{B} \sum_{i=1}^{\lceil \log_2\l(B / \alpha\r) \rceil} 2^i
	\\& < \frac{\alpha n}{B} + 4 \alpha \log_2 B + 2 \alpha
\end{split}
\end{equation*}

\noindent where the first addendum refers to the batches processed unfolding $B$ arcs at a time, while the other addenda refers to the case of smaller batch sizes.
Finally, to get the number of parallel unfoldings during the entire execution, it suffices to sum the quantity above for $\alpha = 1, 2, \dots, 2^{\lceil \log \ell \rceil}$ and we get the desired $O\l(\frac{\ell n}{B} + \ell \log B\r)$.

Now it remains to prove that Algorithm~\ref{batched:pseudocode} uses $O(\ell n)$ operations and space.
The proof is the same as for Algorithm~\ref{bf:alg:pseudocode}, we mainly need to pay attention to lines~\ref{local_copy1} and~\ref{local_copy2} since creating local copies would increase the complexity.
Fortunately, it is sufficient to use the global versions of $A$ and $lost$, store in a list the changes performed on them, then restore their state before terminating \textsc{BuildBatch}.
In fact, we adopted local copies only to make the pseudocode clearer.
Moreover, since the \textsc{BuildBatch} can temporarily skip some vertices (according to the local copy of $lost$) that may be re-included later after the parallel unfold, we cannot employ the linear-space selection described in Section~\ref{sec:implementation_details}.
In this case, we further need to associate to each node $u$ the set (hash table) of all arcs $(u, \cdot) \in E$ unfolded by the algorithm, so to skip the ones already unfolded, in constant time.
The solution proposed in Section~\ref{sec:implementation_details}, which employs the cursors $p_1$ and $p_2$ to decide the arcs to unfold, properly extended with this check, guarantees $O(\ell n)$ time and space.
\end{proof}

\subsubsection*{Implementation Details}
Algorithm~\ref{batched:pseudocode} could not use all the comparisons that are available in a single batch, because of the batch size halving (row~\ref{batched:halving}) or because the brute force call
(row~\ref{batched:brute_force}) involves a number of arcs that is not divisible for the batch size.
For this reason, we employ a simple heuristic to exploit each batch the most, which applies when employing the hash table to store the results of the arc lookups (Section~\ref{sec:implementation_details}).
In detail, we add new arcs to the batch, deterministically, each time Algorithm~\ref{batched:pseudocode} asks to unfold a partially filled batch.
We use an heap data structure to get the node with the smallest number of comparisons lost that still has unfolded arcs, then we add to the batch the remaining unfolded arcs, in the order they appear, until the batch becomes full.
If all node's arcs are added and the batch is still non-full, then the previous operation is repeated until either the batch becomes full or all arcs have been unfolded.


\section{Experiments}
\label{bf:sec:experiments}
In this section, we present a comprehensive experimental assessment of the proposed algorithms on a Question Answering task.
In detail, we focus on passage ranking that aims at selecting, given a question, the most relevant among a set of textual passages answering the question.
To this end, we employ an existing state-of-the-art pairwise model that works by comparing two results at a time and by selecting the winners of the induced round-robin tournament.
In this scenario, the proposed algorithms aim to find the tournament champions by reducing the number of pairwise comparisons, i.e., arc lookup, performed using the ML model.
In the following, we first describe the experimental setting, then we evaluate the proposed algorithms in terms of number of comparisons and speedup of the ranking process.

\subsubsection*{Dataset}
For the the assessment we employ the Microsoft MAchine Reading COmprehension dataset (MS\,MARCO)~\cite{DBLP:conf/nips/NguyenRSGTMD16}.
It is a large scale dataset for Question Answering and consists of approximately $1$ million anonymized questions sampled from the Bing search query logs and about $9$ million passages extracted from web pages.
The development set used for the assessment contains $6,980$ queries having one relevant passage each, on average.

\subsubsection*{Pairwise Model}
Nogueira~\etal{} recently tackled the task of ranking passages by using a three-stage ranking architecture~\cite{DBLP:journals/corr/abs-1910-14424}.
The duoBERT models recently scored among the top-$10$ solutions of the MS\,MARCO Passage Ranking Leaderboard\footnote{\url{https://microsoft.github.io/msmarco/}} and as the first solution whose public code is publicly available\footnote{\url{https://github.com/castorini/duobert}}.
The first stage selects the top-$1000$ results using the fast BM25 algorithm.
The second stage re-ranks these results using a monoBERT neural model~\cite{DBLP:journals/corr/abs-1901-04085}, which ingests the text of a document at a time to classify it as relevant or not.
Lastly, the third stage re-ranks the top-$30$ results of the previous stage by using a duoBERT pairwise model~\cite{DBLP:journals/corr/abs-1910-14424} that classifies all pairs of document's texts to induce a round-robin tournament among the results.
In particular, the two most promising configurations presented in Nogueira~\etal{}~\cite{DBLP:journals/corr/abs-1901-04085} have been tested: duoBERT\textsubscript{PROBABILISTIC} and duoBERT\textsubscript{BINARY}.
The former works by assigning to each document the sum of the probabilities of all comparisons, while the latter rounds these probabilities in $\{0,1\}$ before summing them.

\subsubsection*{Experimental Methodology}
In our experiments, we replicate the full multi-stage pipeline proposed by Nogueira~\etal{} and we apply the proposed algorithms in the last stage of the pipeline, i.e., top-$30$ re-ranking.
In particular, given a query and the set of its top-$30$ passages, each algorithm drives the identification of the champions deciding which pairs of passages to compare using the ML model.
The objective is to retrieve the top passages by reducing the number of pairwise inferences, i.e., arc lookups, performed by the duoBERT models.

We assess the proposed algorithms by measuring the number of comparisons and the time spent by the ML model to perform all inferences.
For fairness, even if our contribution does not regard the effectiveness of the model, we also report the Recall@$k$ metric assessing the fraction of relevant documents captured within the top-$k$ results.

\subsubsection*{Testing Details}
The tests were performed on a machine with sixteen Intel Xeon E5-2630 cores clocked at 2.40GHz, 192GiB RAM, equipped with a NVIDIA TITAN Xp GPU.
The GPU has been used to run the monoBERT and duoBERT models.

\subsection{Experimental Results}
We now present the results of our experimental evaluation.
To ease the discussion, we start by discussing the evaluation in the binary setting for the retrieval of the top-$1$ result (Algorithm~\ref{bf:alg:pseudocode} and its possible implementations).
We then present the results of the proposed algorithms on the problem generalizations, i.e., top-$k$ retrieval, probabilistic setting, and parallel setting (Algorithm~\ref{batched:pseudocode}).

\subsubsection{Asymptotically-optimal Deterministic Algorithm}
Section~\ref{sec:implementation_details} discusses some implementation details to take into account when implementing Algorithm~\ref{bf:alg:pseudocode}.
In particular, there are two orthogonal aspects to consider in the implementation that we want to assess: exploitation of the input order and exploitation of the past arc lookups.
The first aspect exploits the order of the input list when deciding the order of the arc lookups.
Since our inputs consists of $30$ passages that have already been sorted by the second ranking stage, we expect to have more relevant passages in the first positions of the input.
Therefore, it could be desirable to start by performing the comparisons among the more relevant passages coming from the second stage.
The second aspect avoids multiple unfolds of a same arc by storing the arc lookups performed during the tournament.
Therefore, we can easily save time using a little extra space.

\begin{table}[t]
\caption{Average number of inferences of different implementations of Algorithm~\ref{bf:alg:pseudocode} when applied to duoBERT to retrieve the top-$1$ result on the MS\,MARCO dataset. Columns identify whether the implementation exploits the input order, while rows identify whether it exploits the past lookups to avoid multiple unfolds of a same arc.\label{bf:table:table_versions}}
\centering

{
\setlength{\extrarowheight}{0.15em}
\setlength{\tabcolsep}{0.8em}
\begin{tabular}{lrrr@{}l}

\toprule

 & Ignore & Exploit \\
 & input order & input order \\

\midrule

Ignore past lookups
& 126.09 & 125.81 \\

Exploit past lookups
& 76.58 & 64.62 \\

\bottomrule

\end{tabular}
}
\end{table}

\begin{table}[t]
 \caption{Efficiency-Effectiveness performance achieved by monoBERT,
   duoBERT, and duoBERT \& Alg.~\ref{bf:alg:pseudocode} when
   retrieving the top-$1$ result on the MS\,MARCO dataset.}
 \label{bf:table:table1}
 \vspace{0.5em} \centering 

{
\setlength{\extrarowheight}{0.15em}
\setlength{\tabcolsep}{0.25em}
\begin{tabular}{lrrr}

\toprule

Method & Recall@1 & Inferences & Time (s) \\

\midrule

BM25 + monoBERT
& 0.251 & 1000 & 65.91 \\

\midrule

+ duoBERT\textsubscript{BINARY}
& 0.269 & 870 & 57.34 \\

+ duoBERT\textsubscript{BINARY} \& Alg. \ref{bf:alg:pseudocode}
& 0.269 & 65 & 4.26 \\




\bottomrule

\end{tabular}
}
\end{table}

\begin{table*}[t]
  \caption{Efficiency-Effectiveness performance achieved by monoBERT, duoBERT, and duoBERT \& Alg~\ref{bf:alg:pseudocode} when retrieving the top-$k$ results on the MS\,MARCO dataset. The number of inferences of monoBERT and duoBERT is independent of the value of $k$.}
  \label{bf:table:table2}
  \centering

{
\setlength{\extrarowheight}{0.15em}
\setlength{\tabcolsep}{0.6em}
\begin{tabular}{llrrrrrr}

\toprule

\multirow{2}{*}{Method}
& \multirow{2}{*}{Metric}
& \multicolumn{6}{c}{$k$}\\

\cmidrule{3-8}

&
& 1 & 2 & 3 & 4 & 5 & 10 \\

\midrule

\multirow{3}{*}{BM25 + monoBERT}
& Recall
& 0.251 &   0.361 &    0.436 &    0.492 &    0.531 &    0.638\\

& \# Inference
& \multicolumn{6}{c}{1,000}\\

& Time (sec.)
& \multicolumn{6}{c}{65.91}\\

\midrule

\multirow{3}{*}{+ duoBERT\textsubscript{BINARY}}
& Recall
&   0.269 &   0.385 &    0.459 &    0.516 &    0.552 &    0.654\\

& \# Inference
& \multicolumn{6}{c}{870}\\

& Time (sec.)
& \multicolumn{6}{c}{57.34}\\

\midrule

\multirow{4}{*}{\makecell[l]{+ duoBERT\textsubscript{BINARY} \\\& Alg. \ref{bf:alg:pseudocode}}}
& Recall
&   0.269 &   0.385 &    0.459 &    0.516 &    0.552 &    0.654\\

& \# Inference
&  65 &  130 &  234 &  266 &  427 &  711\\

& Time (sec.)
&   4.26 &    8.58 &   15.42 &   17.53 &   28.14 &   46.83\\

& Speedup
& \tspeedup{13.5} & \tspeedup{6.7} & \tspeedup{3.7} & \tspeedup{3.2} & \tspeedup{2.0} & \tspeedup{1.2}\\

\midrule

\multirow{3}{*}{+ duoBERT\textsubscript{PROBABILISTIC}}
& Recall
&   0.266 &   0.385 &    0.460 &    0.514 &    0.550 &    0.653\\

& \# Inference
& \multicolumn{6}{c}{870}\\

& Time (sec.)
& \multicolumn{6}{c}{57.34}\\

\midrule

\multirow{4}{*}{\makecell[l]{+ duoBERT\textsubscript{PROBABILISTIC} \\\& Alg. \ref{bf:alg:pseudocode}}}
& Recall
&   0.266 &   0.385 &    0.460 &    0.514 &    0.550 &    0.653\\

& \# Inference
& 134 & 209 & 291 & 355 & 445 & 732\\

& Time (sec.)
&   8.86 &   13.759 &   19.201 &   23.397 &   29.307 &   48.267\\

& Speedup
& \tspeedup{6.5} & \tspeedup{4.2} & \tspeedup{3.0} & \tspeedup{2.5} & \tspeedup{2.0} & \tspeedup{1.2}\\

\bottomrule
\end{tabular}
}

\end{table*}

We now assess the impact of the two orthogonal implementation aspects described above, which lead to four implementations. 
Table~\ref{bf:table:table_versions} reports the average number of inferences of the different implementations of Algorithm~\ref{bf:alg:pseudocode} when applied to duoBERT to retrieve the top-$1$ result on the MS\,MARCO dataset.
As expected, the two aspects contribute to reduce the average number of inferences.
In particular, we notice that the implementation exploiting the input order is more efficient when used together with the hash table, and that their combination nearly halves the number of inferences of the implementation ignoring both aspects.

Table~\ref{bf:table:table1} reports the performance of the best implementation above, i.e., the one exploiting the input order and the past lookups, within the ranking pipeline proposed by Nogueira~\etal{}
We report Recall@$1$, number of inferences and inference time of all ranking stages.
The first row shows the performance of the first two stages of the ranking pipeline, i.e., BM25 + monoBERT, used here to retrieve the top-$30$ results to re-rank.
It retrieves the correct answer for about $25\%$ of the queries but it requires, on average, about $66$ seconds when applied to the top-$1$,$000$ results returned by BM25.
The second row shows the performance of duoBERT\textsubscript{BINARY} when employed as third stage of the ranking pipeline.
As this model does not guarantee symmetric predictions, each comparison needs two inferences, i.e., $u$ versus $v$ and $v$ versus $u$; it thus requires $30 \times 29 = 870$ inferences.
duoBERT\textsubscript{BINARY} improves the quality of the returned list with respect to the previous stage as it retrieves the correct answer for about $27\%$ of the queries.
However, we want to highlight that this third stage almost doubles the running time as it require about $57$ seconds that must be added to the $66$ seconds required by the first two stages, i.e., BM25 + monoBERT.
The third row of Table \ref{bf:table:table1} shows the performance of the third stage when employing Algorithm~\ref{bf:alg:pseudocode}
to decide which pairs of passages to compare using the duoBERT\textsubscript{BINARY} model.
The recall metric is the same as duoBERT\textsubscript{BINARY}.
This result is expected as we proved the algorithm correctness.
On average, this configuration requires about $4$ seconds per query
and it speeds up the ranking process of the third stage of about $13\times$ with respect to the previous configuration.
Moreover, the time cost of the third stage is now negligible with respect to the one of the first two stages.

The average number of inferences required by our approach is about $65$, which is very close to the minimum number of inferences required to solve this problem when the Champion wins all comparisons, i.e., $29 \times 2 = 58$ inferences.
In particular, $95\%$ of the queries are solved with only $50$ comparisons or less, i.e., solved with less than $100$ model inferences.
In addition, we want to highlight that if we apply the algorithm to a symmetric model, we would not need to perform two inferences per comparison, and the algorithm would perform just a few inferences per item. 

\subsubsection{Top-\ensuremath{k} Retrieval and Probabilistic Version}
Table~\ref{bf:table:table2} reports the performance of Algorithm~\ref{bf:alg:pseudocode} in the top-$k$ retrieval task, both in the binary and the probabilistic settings.
As before, we report Recall@$k$, for $k$ in $\{1,2,3,4,5\}$, number of inferences, and inference time of all ranking stages.
The first row shows the performance of the first two stages of the ranking pipeline introduced by Nogueira~\etal{}, i.e., BM25 + monoBERT.
The second and fourth rows show the performance of duoBERT\textsubscript{BINARY} and duoBERT\textsubscript{PROBABILISTIC} when employed as third stage of the ranking pipeline.
The two configurations require the same number of inferences, i.e., $30 \times 29 = 870$, and the same inference time, as the underlying model is the same.
The binary configuration shows a slightly higher recall than the probabilistic one.
Both the versions improve the recall of the previous ranking stage, thus confirming that tournaments are a good modeling of this problem.
The third and fifth rows show the performance of these models when employing Algorithm~\ref{bf:alg:pseudocode} to perform the tournament among the top-$30$ results of each query.
In both cases the recall is preserved, as the algorithm is correct.
The proposed algorithm speeds up the ranking process from $13\times$ to $2\times$ in the binary setting and from $6\times$ to $2\times$ in the probabilistic setting, for $k$ ranging from $1$ to $5$.
Remark that Algorithm~\ref{bf:alg:pseudocode} obtains excellent results in the top-$1$ retrieval task of both settings.

\begin{table}[t!]
  \caption{Average values of $\ell_k$ when varying $k$ and the tournament type.}
  \label{bf:table:table_values_of_ell}
  \vspace{0.5em} \centering 

{
\setlength{\extrarowheight}{0.15em}
\setlength{\tabcolsep}{0.5em}
\begin{tabular}{lrrrrrr}

\toprule

\multirow{2}{*}{Tournament Type} & \multicolumn{6}{c}{$k$}\\
& 1 & 2 & 3 & 4 & 5 & 10 \\

\midrule

Binary
& 0.05 & 1.09 & 2.13 & 3.15 & 4.18 & 9.19 \\


Probabilistic
& 0.78 & 1.77 & 2.78 & 3.78 & 4.78 & 9.58 \\

\bottomrule

\end{tabular}
}
\end{table}

Taking into account that $\ell_k$, i.e., the number of matches lost by the $k$-th result, drives the time complexity of our algorithm, we report in Table~\ref{bf:table:table_values_of_ell} the different values of $\ell_k$ when varying $k$ and the tournament type, i.e., binary or probabilistic.
The table shows that, on this dataset, $\ell_k$ rapidly increases as $k$ grows and that $\ell_k$ is always higher in the probabilistic setting than in the binary setting.
Indeed, in practice, the number of inferences performed by our algorithm rapidly increases as $k$ grows and that the speedups achieved in the probabilistic setting are always smaller than the ones achieved in the binary setting (Table~\ref{bf:table:table2}).


\subsubsection{Parallel (Batched) Version}
Table~\ref{bf:table:table_parallel} reports the performance of Algorithm~\ref{batched:pseudocode} in the parallel setting where the algorithm can unfold a batch of multiple arcs in parallel.
The table reports the number of inferences and the inference time of all ranking stages, for values of batch size between $2$ and $256$ when retrieving the top-$1$ result on the MS\,MARCO dataset.
The Recall@$1$ metric is not reported as the correctness of the algorithm guarantees that the effectiveness does not change with the batch size.
Indeed, Recall@$1$ is always close to $27\%$ as in the non-parallel setting.
The first row shows the performance of the first two stages of the ranking pipeline, i.e., BM25 + monoBERT, while the second row shows the performance of the third stage, i.e., duoBERT\textsubscript{BINARY}.
The number of batch inferences linearly decreases when increasing the batch size for both configurations, as we can unfold more arcs in parallel per batch.
For instance, with a batch size of $64$, we can perform $64$ inferences at a time and the full round-robin tournament requires only $\lceil 870/64 \rceil = 14$ rounds to perform all inferences.
The third row shows the performance of duoBERT\textsubscript{BINARY} used as third stage when employing Algorithm~\ref{batched:pseudocode} to perform the (batched) tournament among the top-$30$ results of each query.
Our algorithm speeds up the ranking from $13\times$ to $3\times$ for batch size ranging from $2$ to $64$.
As expected, the speedup decreases when increasing the batch size as the number of results involved in the tournament is very limited.
Indeed, the algorithm can accurately unfold only one arc for each alive vertex (Algorithm~\ref{batched:pseudocode}, set $A$); it then fills the batch with a simple heuristic that explores all arcs of just a few promising vertices (as described in the ``Implementation Details'' subsection of Section~\ref{bf:sec:batched}).
Therefore, as the batch size becomes bigger than the number of results, i.e., $30$ in our setting, the choices of the algorithm become less oriented.
Nevertheless, Algorithm~\ref{batched:pseudocode} speeds up the ranking of duoBERT\textsubscript{BINARY} for all the values of batch size tested.

\begin{table*}[t!]
  \caption{Efficiency of parallel (batched) implementations of monoBERT,
    duoBERT, and duoBERT \& Alg~\ref{batched:pseudocode} when
    retrieving the top-1 result on the MS\,MARCO dataset.}
  \label{bf:table:table_parallel}
  \vspace{0.5em} \centering 

{
\setlength{\extrarowheight}{0.15em}
\setlength{\tabcolsep}{0.8em}
\begin{tabular}{lrrrrrrrrr}

\toprule

\multirow{2}{*}{Method}
& \multirow{2}{*}{Metric}
& \multicolumn{8}{c}{Batch Size}\\

&&     2   &     4   &     8   &    16  &    32  &    64  &     128 &     256 \\

\midrule

\multirow{2}{*}{BM25 + monoBERT}
& Inferences &  500 &  250 &  125 &  63 &  32 &  16 &   8 &    4 \\
& Time (s) &   32.95 &   16.48 &    8.24 &   4.15 &   2.11 &   1.05 &   0.53 &    0.26 \\

\midrule

\multirow{2}{*}{+ duoBERT\textsubscript{BINARY}}
& Inferences &  435 &  218 &  109 &  55 &  28 &  14 &   7 &    4 \\
& Time (s) &   28.67 &   14.37 &    7.18 &   3.62 &   1.85 &   0.92 &   0.46 &    0.26 \\






\midrule


\multirow{3}{*}{\makecell[l]{+ duoBERT\textsubscript{BINARY} \\\& Alg. \ref{batched:pseudocode}}}
& Inferences &   33 &   23 &   14 &   8 &   5 &   4 &    4 &    4 \\

& Time (s)
&    2.14 &    1.54 &    0.93 &   0.55 &   0.31 &   0.28 &    0.26 &    0.25 \\

& Speedup &   \tspeedup{13.4} &    \tspeedup{9.3} &    \tspeedup{7.7} &   \tspeedup{6.6} &   \tspeedup{5.9} &   \tspeedup{3.3} &    \tspeedup{1.7} &    \tspeedup{1.0} \\

\bottomrule

\end{tabular}
}
\end{table*}


\section{Conclusion}
\label{bf:sec:summary}
We addressed the problem of how to efficiently solve the retrieval of the top-$1$ result when employing pairwise machine learning classifiers.
We mapped it to the problem of finding champions in tournament graphs by minimizing the number of arc lookups, i.e., the number of comparison done through the classifier.
We showed that, given the number $\ell$ of matches lost by the champion, $\Omega(\ell n)$ arc lookups are required to find a champion, and generalized this statement for randomized algorithms that are only correct with some constant probability.
Then, we presented an asymptotically optimal deterministic algorithm that solves the problem and matches the lower bound without knowing $\ell$.
We also turned our attention to three natural variants of the original problem, and showed algorithms that solve them.
First, we solved the problem of finding all the top-$k$ players simultaneously.
Second, we considered a probabilistic tournament in which any cell of the adjacency matrix contains a probability, and achieved the same performance in that more general case.
Third, we supposed we were able to probe $B$ adjacency matrix cells in parallel and achieved a linear (and thus asymptotically optimal) speedup.
Finally, we experimentally evaluated the proposed algorithms to speed-up a state-of-the-art solution for ranking on public data.
Results show that we are able to speed up the retrieval of the top-$1$ result of up to $13\times$ in the classic binary setting.
We also evaluated the three variants of the original problem and we showed that our proposals speeds-up the retrieval from $13\times$ to $2\times$ for $k$ ranging from $1$ to $5$ in the binary setting (first variant) and from $6\times$ to $2\times$ for the same range of $k$ in the probabilistic setting (second variant).
In the parallel setting (third variant), our proposal consistently speeds up the retrieval of the top-$1$ result for all the values of batch size tested.

As future work, we intend to investigate three main research directions.
On the theoretical side, it would be interesting to characterize the leading constant in the complexity of finding the Copeland winner to better compare the lower bounds and the proposed algorithms.
On a more applied side, it is worth investigating heuristics to increase the speed up of our algorithms while retaining their theoretical performance.
Lastly, it would be also interesting to investigate the dependency between the number of arc lookups performed by our algorithms and the probability distribution of the graph arcs, so to link the complexity to the data at hand.

\ifCLASSOPTIONcaptionsoff
\newpage
\fi

\bibliographystyle{plain}
\bibliography{biblio}

\begin{IEEEbiography}[{\includegraphics[width=1in,height=1.25in,clip,keepaspectratio]{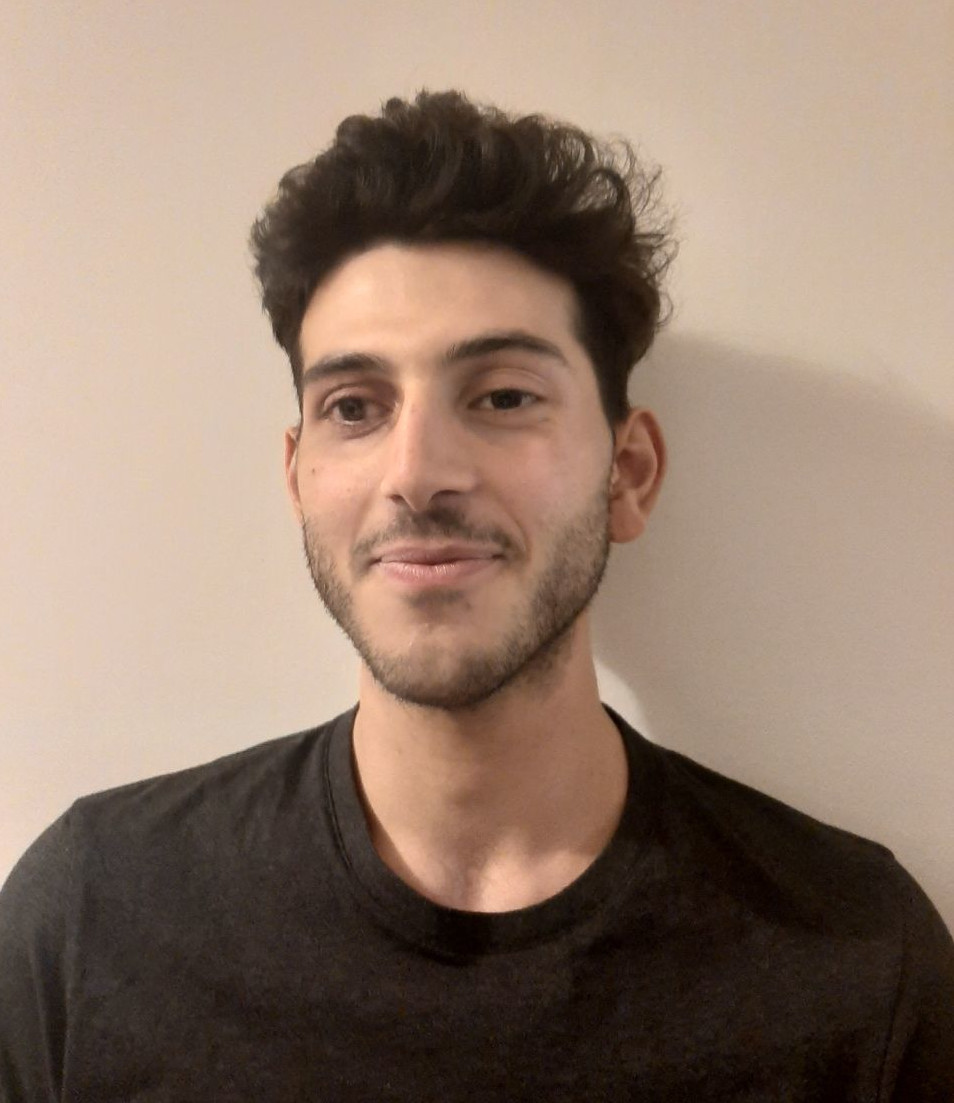}}]{Lorenzo Beretta} received his master degree from University of Pisa and Scuola Normale Superiore in 2020. He is a PhD student at the University of Copenhagen, within the BARC research centre. His research interest is in discrete algorithms: in particular, he is interested in geometrical optimization problems, sublinear algorithms and hashing.
\end{IEEEbiography}

\begin{IEEEbiography}[{\includegraphics[width=1in,height=1.25in,clip,keepaspectratio]{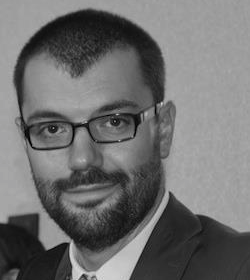}}]{Franco Maria Nardini} received the PhD degree from the University of Pisa in 2011. He is a senior researcher with the National Research Council of Italy. His research interests focus on web information retrieval and machine learning. He authored more than 70 papers in peer-reviewed international journal and conferences. He received the ACM SIGIR 2015 Best Paper Award, ECIR 2014 Best Demo Paper Award, and the ECIR 2022 Industry Impact Award. For more information: \href{http://hpc.isti.cnr.it/~nardini}{http://hpc.isti.cnr.it/~nardini}.
\end{IEEEbiography}

\begin{IEEEbiography}[{\includegraphics[width=1in,height=1.25in,clip,keepaspectratio]{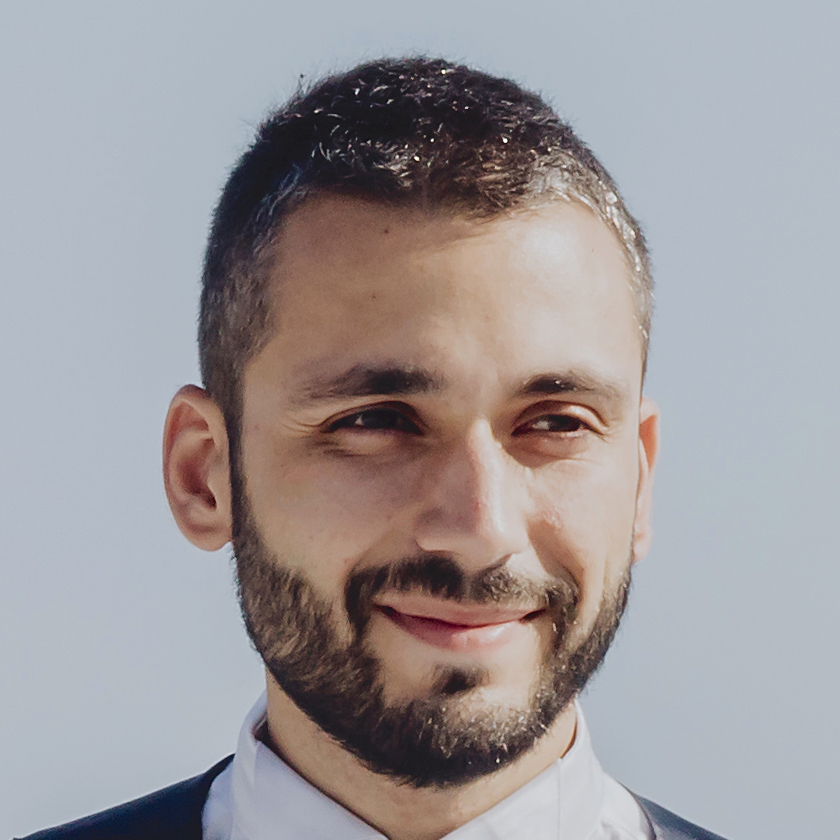}}]{Roberto Trani} received the PhD degree in computer science from the University of Pisa, in 2020. He is a postdoctoral research fellow of the National Research Council of Italy, within the High Performance Computing Laboratory, and his research interests are machine learning, algorithms, information retrieval, and high performance computing.  For more information: \href{http://hpc.isti.cnr.it/roberto-trani}{http://hpc.isti.cnr.it/roberto-trani/}.
\end{IEEEbiography}

\begin{IEEEbiography}[{\includegraphics[width=1in,height=1.25in,clip,keepaspectratio]{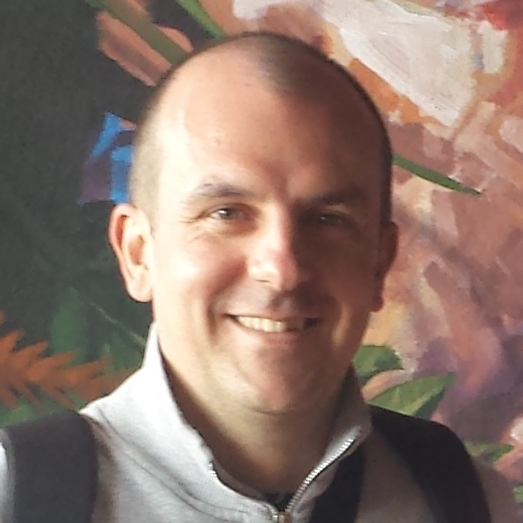}}]{Rossano Venturini} received the PhD degree from the University of Pisa in 2010.  He is an associated professor at the Computer Science Department of the University of Pisa. His research interests are mainly focused on the design and the analysis of algorithms and data structures for indexing and searching large collections.  He received two Best Paper Awards at ACM SIGIR in 2014 and 2015. For more information: \href{http://pages.di.unipi.it/rossano}{http://pages.di.unipi.it/rossano}.
\end{IEEEbiography}
\end{document}